\documentclass[a4paper,11pt]{article}
\usepackage{amsthm,amsfonts,amsmath,verbatim}

\DeclareMathOperator{\card}{card}
\DeclareMathOperator{\alf}{alph}

\newcommand{\PAL}{\mathit{PAL}}
\newcommand{\PER}{\mathop{\textit{PER}}\nolimits}
\newcommand{\Stand}{\mathop{\textit{Stand}}\nolimits}

\newcommand{\fns}{\footnotesize}

\newtheorem{thm}{Theorem}[section]
\newtheorem{prop}[thm]{Proposition}
\newtheorem{cor}[thm]{Corollary}
\newtheorem{lemma}[thm]{Lemma}

\newtheorem{example}[thm]{Example}

\title{Some extremal properties of the Fibonacci word}
\author{Aldo de Luca
\medskip\\
\fns Dipartimento di Matematica e Applicazioni ``R.~Caccioppoli''\\
\fns Universit\`a degli Studi di Napoli Federico II\\
\fns Via Cintia, Monte S.~Angelo, I-80126 Napoli, Italy}

\begin{document}
\maketitle
\begin{abstract} We prove that the Fibonacci word $f$ satisfies among all characteristic Sturmian words, three interesting extremal properties.  The first  concerns  the length and the second the minimal period of its palindromic prefixes. Each of these two properties characterizes $f$ up to a renaming of its letters. A third property concerns the number of occurrences of the letter $b$ in its palindromic prefixes. It characterizes uniquely $f$ among all characteristic Sturmian words having the prefix $abaa$.

\vspace{2 mm}

\end{abstract}

\section{Introduction}

Words are finite or infinite sequences of elements, called letters, taken from a finite set called alphabet. In the combinatorics of infinite words 
 the Fibonacci  word is very famous  since it satisfies a great number of  beautiful  properties which are of a paramount interest
both from the theoretical and the applicative point of view.

As is well known, the Fibonacci word $f$ can be defined in several different ways. For instance,
$f$ is the fixed point $\varphi^{\omega}(a)$ of the Fibonacci morphism $\varphi: \{a, b\}^*\rightarrow  \{a, b\}^*$ defined by $\varphi(a)= ab$ and $ \varphi(b)=a$. The name Fibonacci given to $f$ is due to the fact that  $f$ is the limit sequence of  the infinite sequence $(f_n)_{n\geq -1}$ of finite words recursively defined as 
$$f_{-1}=b,  f_0=a, \ \mbox{and} \  f_{n+1}= f_{n}f_{n-1} \  \mbox{for} \ \  n\geq 0.$$
 For any $n\geq -1$ one has
$|f_n| = F_n$ where $(F_n)_{n\geq -1}$ is the Fibonacci numerical sequence:
$$ F_{-1} = F_0 = 1 \ \mbox{and} \ \  F_{n+1} = F_n + F_{n-1} \  \mbox{for} \  n\geq 0.$$

 The Fibonacci word is a paradigmatic example of Sturmian word. As is well known, Sturmian words are infinite words  over a binary alphabet of great interest in combinatorics on words for the many applications in Algebra, Number theory, Physics, and Computer Science. 
 
 Several different but equivalent definitions of Sturmian words exist (see, for instance,  \cite [Chap.~2]{LO2}).
  A  Sturmian word can be defined in a purely combinatorial way as  an infinite sequence of letters such that for any integer $n\geq 0$, the number of its distinct factors of length $n$ is $n+1$. This is equivalent to say that an infinite word is Sturmian if and only if it is aperiodic and for any $n$ it has the minimal possible of distinct factors of length $n$.

 A geometrical  definition is the following: a  Sturmian word is an infinite word  associated to  the sequence of the cuts (cutting sequence) in a squared-lattice made by a semi-line having a slope which is an irrational number. A horizontal cut is denoted by the letter $b$, a vertical cut by $a$ and a cut with a corner by $ab$ or $ba$. Sturmian words represented by a semi-line starting from the origin are usually called characteristic, or  standard. For any Sturmian word there exists a characteristic Sturmian word having the same set of factors. The Fibonacci word is the characteristic Sturmian word having a slope equal to the golden ratio $g= \frac{\sqrt 5 -1}{2}$.
 
 In many cases, the Fibonacci word $f$ satisfies among all infinite words of a given class, some extremal properties in the sense that   some quantity is maximal or minimal for $f$ (see, for instance \cite{CdL0, CdL,adL81,MRS}, and the overview \cite{JC}). A special case
 of great interest is when the class of infinite words is formed by all characteristic Sturmian words and the extremal property is satisfied only  by the Fibonacci word $f$ and by $E(f)$, where $E$ is the automorphism of $\{a, b\}^*$  interchanging the letter $a$ with the letter $b$. In this way one obtains a characterization of $f$, up to a renaming of the letters, inside the class of characteristic Sturmian words.
 
 Some of these latter extremal properties are strictly related to
a simple construction of characteristic Sturmian words, due to the author \cite{deluca}. It is based on an operator definable in any free monoid $A^*$ and called right-palindromic closure,
which associates to each word $w\in A^*$ the shortest palindrome of $A^*$ having $w$ as a prefix. Any given word $v\in A^*$  can suitably `direct' subsequent  iterations of  the preceding operator  according to the sequence of letters  in $v$ as follows: at each step, one concatenates the next letter of $v$ to the right of the already constructed palindrome and then takes the right palindromic closure. Thus starting with any directive word $v$ one generates a palindrome $\psi(v)$. The map $\psi$,  called palindromization map, is injective; the word $v$ is called the directive word of $\psi(v)$. 

Since for any
$u,v\in A^*$, $\psi(uv)$ has $\psi(u)$ as a prefix, one can extend the map $\psi$ to right infinite words
$x\in A^{\omega}$ producing an infinite word $\psi(x)$. It has been proved in \cite{deluca} that in the case  of a binary alphabet ${\cal A } = \{a, b\}$  if each letter of ${\cal A }$ occurs infinitely often in $x$, then one can generate all characteristic Sturmian words\footnote{The palindromization map $\psi$ has been extended to   infinite words over an arbitrary alphabet $A$ by X. Droubay, J. Justin, and G. Pirillo in  \cite{DJP}, where  the family  of {\em  standard episturmian words} over $A$ has been introduced. Some further extensions and generalizations of $\psi$ are in \cite{adlADL, adlADL1}. An extension of $\psi$ to free group $F_2$ was given by C. Reutenauer  in \cite{KREU}.}. Moreover, $\psi({\cal A}^*)$ coincides with the set of the palindromic prefixes of all characteristic Sturmian words. These words can be also defined in a purely combinatorial way by an extremal property  closely related to Fine and Wilf's periodicity theorem \cite{FW}; they are usually c!
 alled also central words since they play a central role in Sturmian words theory. In Section~\ref{sec:three} some remarkable structural properties of central words relating them to finite standard words and to Christoffel words are briefly presented.

A  central word is of order $n$ if its directive word is of length $n$.
In \cite{CdL} we proved that the Fibonacci word $f= \psi((ab)^{\omega})$ is the only characteristic Sturmian word, up to a renaming of the letters,  whose palindromic prefixes $w$ of any order  are harmonic, that is  the minimal period $\pi(w)$ of $w$ satisfies the condition  $\pi^2(w) \equiv \pm 1$ (mod $|w|+2)$.

The main results of the paper are three theorems (cf.~Theorems \ref{thm:main}, \ref{thm:pmain0}, and \ref{thm: Stuca}), somehow related to each other, showing the following  extremal properties of $f$. Theorem \ref{thm:main} states that  a characteristic Sturmian word $s$ has the palindromic prefixes of any order of maximal length if and only if  $s=f$ or $s= E(f)$, where $E$ is the automorphism of $\{a,b\}^*$
interchanging the letter $a$ with $b$.  Similarly, Theorem~\ref{thm:pmain0} states that, up to a renaming 
of the letters, the Fibonacci word is the only characteristic Sturmian word whose palindromic prefixes of any order have a  maximum value of the  minimal period. Theorem~\ref{thm: Stuca} shows that a characteristic Sturmian word beginning with the letter $a$  has the palindromic prefixes of any order  with the maximal number of occurrences of the letter $b$  if and only if $s = f$ or $s$ has the directive word 
$(ab^2)(ab)^{\omega}$. Hence, this extremal property  characterizes uniquely $f$ among all characteristic Sturmian words having the prefix $abaa$.

The proof of these  theorems is given in Section \ref{sec:four} by using techniques of combinatorics on words and three extremal properties  of central words which are prefixes of the Fibonacci word concerning their length (cf.~Theorem \ref{thm:main00}),
their minimal period (cf.~Theorem \ref{thm:pmain}), and the number of occurrences of the letter $b$
(cf.~Theorem \ref{thm:comp000}).

In Section \ref{sec:five} we consider the arithmetization of Sturmian words theory obtained by representing the directive words of central words, as well as  of characteristic Sturmian words, by sequences of integers (integral representations). In this setting continued fractions and continuants  associated to these numerical sequences play a relevant role.  We show that Theorem~\ref{thm:main00} is equivalent to a  property of continuants (cf.~Theorem \ref{thm:cf0101}) and a direct  proof of this latter result  is also given. Moreover, we show that also  Theorem~\ref{thm:pmain} can be derived from
 Theorem ~\ref{thm:cf0101} by using a suitable expression of the minimal periods of  central words in terms of  continuants.

\section{Preliminaries}

\subsection{Notation and preliminary definitions}

In the following ${\cal A}$  will denote a binary alphabet ${\cal A}=\{a,b\}$  and ${\cal A}^*$  the \emph{free monoid} generated by ${\cal A}$. 
The elements $a$ and $b$ of ${\cal A}$ are usually called \emph{ letters} and those of ${\cal A}^*$ \emph {words}. 
We suppose that  ${\cal A}$ is totally ordered by setting  $a<b$.
 The identity element
of ${\cal A}^*$ is called \emph{empty word} and denoted by $\varepsilon$.  We
set ${\cal A}^+={\cal A}^*\setminus\{\varepsilon\}$.

A word $w\in {\cal A}^+$ can be written uniquely as a sequence of letters 
$w=w_1w_2\cdots w_n$, with $w_i\in {\cal A}$, $1\leq i\leq n$, $n>0$.  The
integer $n$ is called the \emph{length} of $w$ and denoted $|w|$.  The
length of $\varepsilon$ is taken equal to $0$.  For any $w\in {\cal A}^*$ and $x\in {\cal A}$, $|w|_x$
denotes the number of occurrences of the letter $x$ in $w$.
 For any $w\in {\cal A}^*$, $\alf w $ will denote the set of all distinct letters of ${\cal A}$ occurring in $w$.
 
 We consider the map  $\eta: {\cal A}^* \rightarrow {\mathbb Q} \cup \{\infty\}$ defined by $$\eta(\varepsilon)= 1 \ \ \mbox {and} \ \  \eta(w) = \frac{|w|_b}{|w|_a}\ \ \mbox{for} \ \  w\neq \varepsilon.$$ If $|w|_a=0$ and $w\neq \varepsilon$, we assume   $\eta(w) =\frac{|w|_b}{0}= \infty$. For any
$w\in {\cal A}^*$, $\eta(w)$ is called the {\em slope} of $w$.

Let $w\in {\cal A}^*$.  The word $u$ is a \emph{factor} of $w$ if there exist
words $r$ and $s$ such that $w=rus$.  A factor $u$ of $w$ is called
\emph{proper} if $u\neq w$.  If $w=us$, for some word $s$ (resp.,
$w=ru$, for some word $r$), then $u$ is called a \emph{prefix} (resp.,
a \emph{suffix}) of $w$.

   Let  $p$ be a positive integer.  A  word $w=w_1\cdots w_n$, $w_i\in {\cal A}$,
$1\leq i\leq n$, has {\em period} $p$  if the following condition is
satisfied: for any integers $i$ and $j$ such that $1\leq i,j\leq n$,
\[
	\mbox{if }i\equiv j \pmod{p} , \mbox{ then } w_i = w_j.
\]
Let us observe that if a word $w$ has a period $p$, then any non-empty  factor of $w$ has also the period $p$. 
We shall denote by $\pi(w)$ the minimal period of $w$. Conventionally, we set $\pi(\varepsilon)=1$.

We recall the following important periodicity  theorem due to Fine and Wilf \cite{FW}:  If a word $w$ has two periods $p$ and $q$  and $|w|\geq p+q -\gcd(p,q)$, then $w$ admits the period $\gcd(p,q)$.

Let $w=w_1\cdots w_n$, $w_i\in {\cal A}$, $1\leq i\leq n$.  The
\emph{reversal}, or {\em mirror image}, of $w$ is the word $w^{\sim}= w_n\cdots w_1$.  One
defines also $\varepsilon^{\sim}=\varepsilon$.  A word is called
\emph{palindrome} if it is equal to its reversal.  We shall denote by
$\PAL$ the set of all palindromes on the
alphabet ${\cal A}$.

A  right-infinite word $x$, or simply \emph{infinite word}, over the alphabet ${\cal A}$   is just an infinite sequence of letters:
$$x=x_1x_2\cdots x_n\cdots \text{ where }x_i\in {\cal A},\,\text{ for all } i\geq 1\enspace.$$
For any integer $n\geq 0$, $x_{[n]}$ will denote the prefix $x_1x_2\cdots x_n$ of $x$ of length $n$.
A factor of $x$ is either the empty word or any sequence  $x_i\cdots x_j$ with $i\leq j$. The set of all infinite words over ${\cal A}$ is denoted by ${\cal A}^{\omega}$.

For all definitions and notation concerning words not explicitly given in the paper, the reader is referred to the book of  M. Lothaire \cite{LO}; for Sturmian words see \cite[Chap.~2]{LO2}.

\subsection{The palindromization map}\label{sec:twotwo}
 We  introduce in
${\cal A}^*$ the operator $^{(+)} : {\cal A}^*\rightarrow \PAL$ which
maps any word $w\in {\cal A}^*$ into the palindrome $w^{(+)}$
defined as the shortest palindrome having the prefix $w$ (cf.~\cite{deluca}).  We call
$w^{(+)}$ the \emph{right palindromic closure of} $w$.  If $Q$ is the
longest palindromic suffix of $w= uQ$, then one has
\[
	w^{(+)}=uQu^{\sim}\,.
\]
Let us now define the map
\[
	\psi: {\cal A}^*\rightarrow \PAL ,
\]
called {\em right iterated palindromic closure}, or simply {\em palindromization map},  as follows: $\psi(\varepsilon)=\varepsilon $ and for all
$v\in {\cal A}^*$, $x\in {\cal A} $,
\[
	\psi(vx)=(\psi(v)x)^{(+)}\,.
\]

\vspace{2 mm}

\begin{example}\label{ex:pm} {\em Let $v=ab^2a$. One has $\psi(a)=a$, $\psi(ab)= (ab)^{(+)}=aba$,
$\psi(ab^2)= ababa$, and $\psi(v)=(ababaa)^{(+)}= ababaababa$.}
\end{example}

The following proposition summarizes some  noteworthy  properties of  the palindromization map
 (cf., for instance \cite{deluca,DJP}):
\begin{prop}\label{prop:basicp} The palindromization map $\psi$  satisfies the following
properties:
\begin{itemize}
\item[P1.] The   palindromization map  is  injective.
\item[P2.]  If $u$ is  a prefix of  $v$, then $\psi(u)$ is a palindromic prefix (and suffix) of $\psi(v)$.
\item[P3.] If $p$ is a prefix of $\psi(v)$, then $p^{(+)}$ is a prefix of $\psi(v)$.
\item[P4.] Every palindromic prefix of $\psi(v)$ is of the form $\psi(u)$ for some prefix $u$ of $v$.
\item[P5.] The palindromization map $\psi$ commute with  the automorphism $E$ of  ${\cal A}^*$ defined by $E(a)=b$ and
$E(b)=a$, i.e.,
$ \psi\circ E = E\circ \psi.$
\item[P6.] For every $v\in {\cal A}^*$, $|\psi(v)|= |\psi(v^{\sim})|$.
\end{itemize}
\end{prop}
 For any  $w\in \psi({\cal A}^*)$ the unique word  $v$ such that  $\psi(v)=w$ is called the \emph{directive word} of $w$.
  The directive word $v$ of $w=\psi(v)$ can be read from $w$
just by taking the subsequence of $w$ formed by all letters immediately following all proper palindromic
prefixes of $w$.

For any $x\in {\cal A}$ let  $\mu_x$ denote  the injective endomorphism of ${\cal A}^*$
$$\mu_x: {\cal A}^*\rightarrow {\cal A}^*$$ 
defined by
\begin{equation}\label{eq:endo}
\mu_x(x)=x, \ \ \mu_x(y)= xy, \, \, \mbox{for}  \, \,  y\in {\cal A}\setminus \{x\} .
\end{equation}
If $v=x_1x_2\cdots x_n$, with $x_i\in {\cal A}$, $i=1,\ldots, n$, then we set:
$$ \mu_v=\mu_{x_1}\circ \cdots \circ \mu_{x_n}; $$
moreover, if $v=\varepsilon$,  $\mu_{\varepsilon}$= id.

The following interesting theorem, proved by J. Justin \cite{J} in the case of an arbitrary alphabet, relates the palindromization map to
morphisms $\mu_v$.

\begin{thm}\label{thm:J} For all  $ v,u\in {\cal A}^*$,
$$ \psi(vu) = \mu_v(\psi(u))\psi(v).$$
In particular, if $x\in{\cal A}$, one has:
$$\psi(xu)= \mu_x(\psi(u))x \ \ \mbox{and} \ \ \psi(vx)=\mu_v(x)\psi(v).$$
\end{thm}

\begin{example} {\em Let $v=ab^2a$.  One has (see Example \ref{ex:pm}) $\psi(v)=ababaababa$ and
$\psi(av)=\mu_a(\psi(v))a= aabaabaaabaabaa$.}
\end{example}

One can extend
  $\psi$  to ${\cal A}^{\omega}$ as follows: let  $x\in {\cal A}^{\omega}$ be an infinite word
$$ x = x_1x_2\cdots x_n\cdots, \ \  \ x_i\in {\cal A}, \ i\geq 1.$$
Since by property P2 of Proposition \ref{prop:basicp}   for all $n$, $\psi(x_{[n]})$ is a proper prefix of  $\psi(x_{[n+1]})$,  we  can define  the infinite word $\psi(x)$ as:
$$ \psi(x) = \lim_{n\rightarrow \infty} \psi(x_{[n]}).$$
The extended map $\psi: {\cal A}^{\omega}\rightarrow {\cal A}^{\omega}$ is injective. The word $x$ is called the {\em directive word} of $\psi(x)$. It has been proved in \cite{deluca} that the word $\psi(x)$ is a {\em characteristic Sturmian word}  if and only if both the letters $a$ and
$b$ occur infinitely often in the directive word  $x$. From property P4 of Proposition \ref{prop:basicp} one easily derives that $\psi({\cal A}^*)$ is equal to the set of the palindromic prefixes of all characteristic Sturmian words.

 \begin{example} {\em  Let ${\cal A}=\{a,b\}$. If $x = (ab)^{\omega}$, then the characteristic Sturmian word $\psi( (ab)^{\omega})$ having the directive word $x$ is the {\em Fibonacci word}
 $$ f = abaababaabaababaababaabaa\cdots$$}
 \end{example}

 \section{Central, standard, and Christoffel words}\label{sec:three}
 
 In this section we consider  three noteworthy classes of finite words
 called {\em central, standard}, and {\em Christoffel words}  which are closely interrelated and are very important in the combinatorics of Sturmian words  as they satisfy
 remarkable structural properties and, moreover, can  be regarded as a finite counterpart of Sturmian sequences. 

A word $w$ is called central if $w$ has  two
periods $p$ and $q$ such that $\gcd(p,q)=1$ and $|w|= p+q-2$. Thus a word is central if it is a power of a single letter or is a word of maximal length for which the theorem of Fine and Wilf does not apply.
The set of central words, usually denoted by $\PER$,  was introduced in~\cite{DM} where its main properties
were studied.  It has been proved that $\PER$ is equal
to the set of the palindromic prefixes of all characteristic Sturmian words, i.e.,
$$ \PER = \psi({\cal A}^*).$$ 
The term \emph{central} was given by J. Berstel and P. S\'e\'ebold in~\cite [Chap.~2]{LO2} to emphasize  the central role that these words play in Sturmian words theory. 

We say that a central word $w$ is of {\em order} $n$ if its directive word has length $n$. As proved in \cite{DM} the number of central words of order $n$ is $\phi(n+2)$ where
$\phi$ is the totient Euler function.
The following remarkable structural characterization of central words holds~\cite{deluca,CdL}:

\begin{prop}
	\label{Prop:uno}
	A word $w$ is central if and only if $w$ is a power of a single
	letter of $\cal A$ or it satisfies the equation:
	\[
		w=w_1abw_2=w_2baw_1
	\]
	with $w_{1},w_{2}\in {\cal A}^*$.  Moreover, in this latter case, $w_1$
	and $w_2$ are uniquely determined central words, $p=|w_1|+2$ and $q=|w_2|+2$ are
	coprime periods of $w$, and $\min\{p,q\}$ is the minimal period of
	$w$.
\end{prop}

Another important family of finite words, strictly related to central words, is the class of finite standard words. In fact,  characteristic  Sturmian words can be equivalently defined in the following
way.  Let $c_1,\ldots,c_n,\ldots$ be any sequence of
integers such that $c_1\geq 0$ and $c_i>0$ for $i>1$.  We define,
inductively, the sequence of words $(s_n)_{n\geq -1}$, where
\[
	s_{-1}=b,\ s_0=a, \ \mbox{ and } \ 
	s_{n}=s_{n-1}^{c_{n}}s_{n-2} \ \mbox{ for } \ n\geq 1\,.
\]
Since for any $n\geq 0$, $s_n$ is a proper prefix of $s_{n+1}$, the sequence $(s_n)_{n\geq -1}$ converges to a limit $s$ which is a
characteristic Sturmian word (cf.~\cite{LO2}).  Any characteristic
Sturmian word is obtained in this way. The sequence $(c_1, c_2, \ldots, c_n, \ldots)$ is called
the {\em directive numerical sequence} of $s$. The Fibonacci word is obtained
when  $c_i=1$ for  $i\geq 1$. 

We shall denote by $\Stand$
the set of all the words $s_n$, $n\geq -1$ of any sequence
$(s_n)_{n\geq -1}$.  Any word of $\Stand$ is called \emph{finite standard  word}, or
simply {\em standard word}.

The following remarkable relation existing between standard and central
words, has been proved in~\cite{DM}:
$$ \Stand = {\cal A} \cup \PER \{ab, ba\}.$$
More precisely, the following holds (see, for instance~\cite[Propositions 4.9 and 4.10]{SC}):
\begin{prop} \label{prop:standStu}Any standard word different from a single letter can be uniquely expressed
as $\mu_v(xy)$ with $\{x,y\}=\{a,b\}$ and $v\in {\cal A}^*$. Moreover, one has:
$$\mu_v(xy)= \psi(v)xy.$$
\end{prop}

\vspace{2 mm}

Let us set for any  $v\in {\cal A}^*$ and $x\in{\cal A}$,  $p_x(v)= |\mu_v(x)|$. From Justin's formula one derives (see, for instance,~\cite{adlZ2})  that $p_x(v)$ is the minimal period of $\psi(vx)$ and then a period of $\psi(v)$. Moreover, $p_x(v)= \pi(\psi(v)x)$, $\gcd(p_x(v), p_y(v))=1$, 
\begin{equation}\label{eq:perpsi}
\pi(\psi(v))= \min \{p_x(v), p_y(v) \},
\end{equation}
 and from Proposition \ref{prop:standStu},  $$|\psi(v)|= p_x(v)+p_y(v)-2.$$

Let us now recall  the important notion of {\em Christoffel word} \cite{CFF} (see also~\cite{BLRS,BDR}). Let $p$ and $q$ be positive relatively prime integers such that $n= p+q$. The Christoffel word $w$ of slope $\frac{p}{q}$ is defined as $w= x_1\cdots x_n$ with 
\[ x_i = \left \{ \begin{array}{ll}
a & \mbox{if $ip \mod n > (i-1)p \mod n$} \\
b & \mbox{if $ip \mod n < (i-1)p \mod n$}
\end{array}
\right . \]
for $i=1,\ldots, n$ where $k \mod n$ denotes the remainder of the Euclidean division of $k$ by $n$.
The term slope given to the irreducible fraction $\frac{p}{q}$  is due to the fact that, as one easily derives from the definition,  $p = |w|_b$ and $q= |w|_a$. The words $a$ and $b$ are also Christoffel words with a respective slope $\frac{0}{1}$ and
$\frac{1}{0}$. The Christoffel words of slope $\frac{p}{q}$ with $p$ and $q$ positive integers are called {\em proper Christoffel words}. 

Let us denote by $CH$ the class of Christoffel words. The following important result, proved in \cite{BDL},
shows a basic relation existing between central and Christoffel words:
$$ CH = a \PER b \ \cup \cal A. $$
Hence, there exists a simple bijection of the set of central words onto the set of proper Christoffel words. Any proper Christoffel word $w$ can be uniquely represented as $a\psi(v)b$ for a suitable $v\in {\cal A}^*$. 

Let $<_{lex}$ denote  the lexicographic order of  ${\cal A}^*$ and let  $Lynd$ be the set of Lyndon words \cite{LO} of  ${\cal A}^*$ and $St$ be the set of  (finite) factors of all Sturmian words. The following theorem summarizes some results on Christoffel words proved in \cite{BL,BDL,BDR,adlZ2}.
\begin{thm}\label{thm:Lynd} Let $w=a\psi(v)b$ with $v\in{\cal A}^*$  be a proper Christoffel word. Then the following hold:
\begin{enumerate}
\item[1.] $CH = St \cap Lynd$, i.e., $CH$ equals the set of all  factors of Sturmian words which are Lyndon words.
\item[2.] There exist and are unique two Christoffel words
$w_1$ and $w_2$ such that $w=w_1w_2$. Moreover,  $w_1<_{lex}w_2$, and $(w_1,w_2)$ is the standard factorization of $w$ in Lyndon words. 
\item[3.] If $w$ has the slope $\eta(w)= \frac{p}{q}$, then $|w_1|=p'$,
$|w_2|=q'$, where $p'$ and $q'$ are the respective multiplicative inverse of $p$ and $q$ $(mod \ |w|)$. Moreover,
$p'= p_a(v)$, $q'=p_b(v)$ and $p= p_a(v^{\sim})$, $q=p_b(v^{\sim})$.
\end{enumerate} 
\end{thm}

\begin{example} {\em The Christoffel word having slope $\frac{5}{12}$ is  $$w= aaabaabaaabaabaab =aub,$$
where $u=aabaabaaabaabaa = \psi(a^2b^2a)$ is the central word of length $15$ having  the two
coprime periods $7=p_a(v)$ and $10=p_b(v)$ with $v=a^2b^2a$. The word $w$ can be uniquely factorized as $w=w_1w_2$, where $w_1$ and $w_2$ are the Lyndon words $w_1= aaabaab$ and $w_2= aaabaabaab$. One has $w_1<_{lex} w_2$ with $|w_1|=7=p_a(v)$ and $|w_2|=10=p_b(v)$. Moreover, $w_2$ is the proper suffix of $w$ of maximal length which is a Lyndon word. Finally, 
$\psi(v^{\sim})= \psi(ab^2a^2)= ababaababaababa$, $p_a(v^{\sim})= 5= |w|_b$, $p_b(v^{\sim})=12=|w|_a$ and $|w|_bp_a(v)= 5^. 7=35 \equiv |w|_ap_b(v)= 12^.10=120 \equiv 1 \mod 17$. }
\end{example}

\section{The Fibonacci word}\label{sec:four}

The Fibonacci word $f$  is without doubt the most famous characteristic Sturmian word. As is well known it can be constructed in several differents  ways. As we have seen is Section \ref{sec:twotwo},
$f$ can be generated by the palindromization map $\psi$  from the directive word
$x= (ab)^{\omega}$, i.e., $f=\psi(x)$. In the following we set for any $n\geq 0$
$$ v^{(n)}= x_1\cdots x_n = x_{[n]},$$
so that  $v^{(0)}= \varepsilon$, $v^{(1)}=a$,
$$v^{(n)}= (ab)^{\frac{n}{2}}\  \mbox{ if} \ n \ \mbox{ is even, and} \  v^{(n)}= (ab)^{\lfloor\frac{n}{2}\rfloor}a \mbox{ if} \ n \ \mbox{ is odd}. $$

\begin{thm}\label{thm:main00} Let $n\geq 0$. For any $v\in {\cal A}^{n}$ one has:
$$ |\psi(v)| \leq |\psi(v^{(n)})|,$$
where the equality holds if and only if
$$ v = v^{(n)} \ \ \mbox{or} \ \ v = E(v^{(n)}).$$
\end{thm}
\begin{proof} The proof is by induction on the length $n$ of $v$. The result is trivially true for $n\leq 1$.
For $n=2$ the result is also true since $|\psi(aa)|=|\psi(bb)|=2$, whereas $|\psi(ab)|= |\psi(ba)|=3$.
Let us then suppose that the result is achieved up to the length $n\geq 2$ and prove it for  the length $n+1$.

We can write $v^{(n+1)}= v^{(n)}z$ with $z=a$ if $n$ is even and $z=b$, otherwise. By  Justin's formula (cf.~Theorem \ref{thm:J}) one has:
\begin{equation}\label{eq:fib0}
 \psi(v^{(n+1)})= \psi(v^{(n)}z)= \mu_{ v^{(n)}}(z)\psi(v^{(n)}).
 \end{equation}
From the definition $v^{(n)}= v^{(n-1)} \bar z$ having set $\bar z = E(z)$. Thus, since by (\ref{eq:endo}) $\mu_{\bar z}(z)= \bar z z$, from Proposition \ref{prop:standStu} one has
$$ \mu_{ v^{(n)}}(z)= (\mu_{ v^{(n-1)}}\circ \mu_{\bar z})(z)=  \mu_{ v^{(n-1)}}(\bar z z)= \psi(v^{(n-1)}) \bar z z$$
and replacing in (\ref{eq:fib0}), one derives:
\begin{equation}\label{eq:fib00}
\psi(v^{(n+1)})=\psi(v^{(n-1)}) \bar z z \psi(v^{(n)}).
\end{equation}
Let $v\in {\cal A}^{n+1}$ and write $v= uy$ with $u\in {\cal A}^n$ and $y\in {\cal A}$.  One has by  Justin's formula:
\begin{equation}\label{eq:eq:fib1}
\psi(v)= \psi(uy)= \mu_u(y)\psi(u).
\end{equation}

If $v\in y^*$, i.e., $v=y^{n+1}$, then $\psi(v)= y^{n+1}$. In this case we are done since for $n\geq 1$,
$|\psi(v)|= n+1 <|\psi(v^{(n+1)})|$ (cf.~Lemma \ref{lem:fib000}). Let us then suppose that $\card(\alf v)= 2$. We can write
$u=u'\bar y \zeta$ with $\zeta\in y^*$ and $u'\in {\cal A}^*$. From (\ref {eq:eq:fib1}) and  Proposition \ref{prop:standStu} one has, since $\mu_{\zeta}(y)=y$,
$$\psi(v)= \mu_{u'\bar y}(y)\psi(u)= (\mu_{u'}\circ \mu_{\bar y})(y)\psi(u)= \mu_{u'}(\bar y y)\psi(u)= \psi(u')\bar y y \psi(u).$$
From (\ref{eq:fib00}) and the preceding equation it follows:
\begin{equation}\label{eq:fib3}
|\psi(v^{(n+1)})|- |\psi(v)|= (|\psi(v^{(n)})|- |\psi(u)|)+ (|\psi(v^{(n-1)})|- |\psi(u')|).
\end{equation}
Setting $k= |u'|\leq n-1$ one has $|\psi(v^{(n-1)})|\geq |\psi(v^{(k)})|$, so that
$$|\psi(v^{(n+1)})|- |\psi(v)|\geq (|\psi(v^{(n)})|- |\psi(u)|)+ (|\psi(v^{(k)})|- |\psi(u')|).$$
By induction, $|\psi(v^{(n)})|\geq |\psi(u)|$ and $|\psi(v^{(k)})|\geq |\psi(u')|$ that implies
$$ |\psi(v^{(n+1)})|\geq  |\psi(v)|,$$
which proves the first part of theorem.

If $v=v^{(n+1)}$ or $v=E(v^{(n+1)})$, then $|\psi(v)| = |\psi(v^{(n+1)})|$. Indeed, one has only to observe that in view of property P5 of Proposition \ref{prop:basicp}, $\psi(E(v^{(n+1)}))= E (\psi(v^{(n+1)}))$, so that $ |\psi(E(v^{(n+1)}))|=|\psi(v^{(n+1)})|$.

Conversely, let us  suppose that  $|\psi(v)| = |\psi(v^{(n+1)})|$. From (\ref{eq:fib3}) one derives:
\begin{equation}\label{eq:fib4}
|\psi(v^{(n)})| = |\psi(u)| \ \mbox{and} \  |\psi(v^{(n-1)})| = |\psi(u')|.
\end{equation}
From equation $(\ref{eq:fib4})_2$ one obtains  $k=|u'|=n-1$. Indeed, if $k<n-1$ one would have:
$|\psi(v^{(n-1)})| >|\psi(v^{(k)})|\geq |\psi(u')|$, a contradiction. Hence, from $(\ref{eq:fib4})_2$ one has
$u=u'\bar y$ and $v=uy=u'\bar y y$.

By induction (\ref{eq:fib4}) is satisfied if and only if 
$$ \mbox{ a)} \  u= v^{(n)}  \ \mbox{or} \ \ \mbox{b)} \ u = E( v^{(n)})$$
and
$$ \mbox{ c)} \  u'= v^{(n-1)}  \ \mbox{or} \ \ \mbox{d)} \ u' = E( v^{(n-1)}).$$
Since $u'$ is a non-empty prefix of $u$,  condition $a)$ $\&$ $d)$, as well as $b)$ $\&$ $c)$, is
a contradiction. Indeed, $u'$ would begin with the letter $a$ and with the letter $b$. Thus  (\ref{eq:fib4}) is satisfied if and only if 
$$ u= v^{(n)}  \ \mbox{and} \ \  \ u' =  v^{(n-1)}$$
or
$$   u= E(v^{(n)}) \  \ \mbox{and} \ \  \ u' = E( v^{(n-1)}).$$
In the first case one has:
$$  v^{(n+1)} =  v^{(n)} z= uz=  v^{(n-1)}\bar z z.$$
Moreover, $u=u'\bar y=   v^{(n-1)}\bar y$, so that
$$  v^{(n+1)} =  uz = v^{(n-1)}\bar y z.$$
Hence, $y=z$ and 
$$ v = uy= uz =  v^{(n+1)}.$$
In the second case one has:
$$v = uy= u'\bar y y = E( v^{(n-1)})\bar y y =  E(v^{(n)})y.$$
Thus $E(v^{(n)})= E(v^{(n-1)})\bar y = E(v^{(n-1)}y)$, so that $ v^{(n)}=   v^{(n-1)}y$. Since
 $v^{(n+1)}= v^{(n)}z$, one derives:  $v^{(n+1)}= v^{(n-1)}yz$. This implies $z=\bar y$ and
 $$ v=  E(v^{(n)})y= E(v^{(n+1)}),$$
 which concludes our proof.
\end{proof}
\begin{lemma}\label{lem:fib000} Let $(F_n)_{n\geq -1}$ be the Fibonacci numerical sequence. For all $n\geq 0$ one has:
$$|\psi(v^{(n)})| = F_{n+1} -2.$$
\end{lemma}
\begin{proof} The result is trivial for $n\leq 1$. Indeed, for $n=0$ one has $|\psi(\varepsilon)|=0$ and
$F_1=2$. For $n=1$, $|\psi(a)|=1$ and $F_2=3$. Suppose by induction the result true up to $n$ and prove it for $n+1$. By (\ref{eq:fib00}) one has:
$$|\psi(v^{(n+1)})|= |\psi(v^{(n-1)})| +|\psi(v^{(n)})|+2.$$
Since by induction $|\psi(v^{(n-1)})|= F_n -2$ and $|\psi(v^{(n)})|=F_{n+1}-2$, the result
follows.
\end{proof}
\begin{cor}\label{cor:main}Let $n\geq 0$. For any $v\in {\cal A}^{n}$ one has:
$$ |\psi(v)| \leq F_{n+1}-2,$$
where the equality holds if and only if
$$ v = v^{(n)} \ \ \mbox{or} \ \ v = E(v^{(n)}).$$
\end{cor}
\begin{proof} Immediate from Theorem \ref{thm:main00} and Lemma \ref{lem:fib000}.
\end{proof}
Let us recall (cf.,~Section \ref{sec:three}) that a  palindromic prefix of a characteristic
Sturmian word is of {\em order}  $n$ if its directive word is of length $n$. From Theorem~\ref{thm:main00}  the following extremal property of the Fibonacci word holds:
\begin{thm}\label{thm:main} A characteristic Sturmian word $s$ has the palindromic prefixes of any order of maximal length if and only if  $s=f$ or $s= E(f)$.
\end{thm}
\begin{proof} Let $s= \psi(y)$, with $y=y_1\cdots y_n\cdots$, $y_i\in {\cal A}$, $i\geq 1$, be any characteristic Sturmian word. By Theorem \ref{thm:main00} for any $n\geq 0$,
$$ |\psi(y_1\cdots y_n)| \leq |\psi(v^{(n)})| = |\psi(E(v^{(n)}))|, $$
where $v^{(n)}$  and $E(v^{(n)})$ are respectively the prefixes of $(ab)^{\omega}$ and of $(ba)^{\omega}$ of length $n$. Since  $\psi(v^{(n)})$ and $E(\psi(v^{(n)}))$ are respectively the palindromic prefixes of order $n$ of $f$ and of $E(f)$,  the  `if part'  of theorem follows. 

Let now $s=\psi(y)$ be any characteristic Sturmian word such that for any $n$ and $v\in {\cal A}^n$,
 $|\psi(y_1\cdots y_n)| \geq |\psi(v)|$. In particular, one has  $|\psi(y_1\cdots y_n)| \geq |\psi(v^{(n)})|$.
 By Theorem \ref{thm:main00} it follows that for any $n\geq 0$
 $$|\psi(y_1\cdots y_n)| =  |\psi(v^{(n)})|.$$
 Moreover, the equality  occurs  if and only if  $y_1\cdots y_n = v^{(n)}$ or   $y_1\cdots y_n = E(v^{(n)})$. Since for $n>0$, $v^{(n)}$  begins with the letter $a$ and $E(v^{(n)})$ begins with the letter $b$, it follows that
 either for any $n\geq 0$, $y_1\cdots y_n = v^{(n)}$ or for any $n\geq 0$, $y_1\cdots y_n = E(v^{(n)})$, i.e., $s=f$
 or $s= E(f)$, which concludes the proof.
\end{proof}

Let us introduce in ${\cal A}^*$ the operator  $c$ defined as:
$c(\varepsilon)= \varepsilon$, $c(x)= x$ for any $x\in {\cal A}$, and for $v=uxy$ with $u\in{\cal A}^*$, $x,y\in {\cal A}$, $c(v)=c(uxy)= uyx$. Thus the operator $c$ acting on  words $v$ of length $\geq 2$ changes the suffix  $xy$ of $v$ of length 2 in $yx$. Note that if $x\neq y$, then $c(uxy)= u\bar x\bar y$. For instance, $c(abbaba)= abbaab$.
It is ready verified that the operator $c$ commutes with $E$, i.e.,  $c\circ E = E\circ c$.

The following theorem concerns the minimal periods of the central words having a directive word of any length.

\begin{thm}\label{thm:pmain} For any $n\geq 0$ and $v\in {\cal A}^n$,
$$\pi(\psi(v)) \leq \pi(\psi(v^{(n)})) = F_{n-1},$$
where the maximum is reached if and only if  $v$ is one of the following words:
$$ v^{(n)}, E(v^{(n)}), c(v^{(n)}), \ \mbox{and} \ \  E(c(v^{(n)})).$$
\end{thm}
\begin{proof} The result is trivial for $n=0$.  We first prove that for any $n \geq 0$
$ \pi(\psi(v^{(n+1)})) = F_{n}$. Indeed, setting $v^{(n+1)}= v^{(n)}z$ with $z\in {\cal A}$ one
has, in view of (\ref{eq:fib00}), 
$$\psi(v^{(n+1)})=\psi(v^{(n-1)}) \bar z z \psi(v^{(n)})=\psi(v^{(n)})z\bar z\psi(v^{(n-1)}).$$
From Proposition \ref{Prop:uno} and Lemma \ref{lem:fib000},  one has:
$$ \pi(\psi(v^{(n+1)})) = \min \{ |\psi(v^{(n-1)})|+2, |\psi(v^{(n)})|)+2\} = |\psi(v^{(n-1)})|+2 =F_n.$$
We prove now that for any $v\in{\cal A}^{n+1}$, $\pi(\psi(v)) \leq \pi(\psi(v^{(n+1)})) = F_{n}$. 

Indeed, we can write $v=uy$ with $u\in {\cal A}^{n}$ and $y\in {\cal A}$. If $u=y^n$, then $v=y^{n+1}$ and
$\psi(y^{n+1})= y^{n+1}$ that implies $\pi(y^{n+1})=1\leq F_n$. Let us then suppose $\card(\alf v)=2$.
As we have seen in the proof of Theorem \ref{thm:main00}, we can write
$u=u'\bar y \zeta$ with $\zeta\in y^*$ and $u'\in {\cal A}^*$ having:
$$\psi(v)= \psi(u')\bar y y \psi(u).$$
From Proposition \ref{Prop:uno}, as $|\psi(u')|<|\psi(u)|$, one has
\begin{equation}\label{eq:psiu}
\pi(\psi(v))= |\psi(u')|+2.
\end{equation}
By Theorem \ref{thm:main00} and Lemma \ref{lem:fib000}, $|\psi(u')|\leq |\psi(v^{(|u'|)})|= F_{|u'|+1}-2$. Since $|u'|\leq n-1$ it follows $|\psi(u')|\leq F_n-2$. Hence, from (\ref{eq:psiu}) one obtains that
for all $v\in {\cal A}^{n+1}$,
$\pi(\psi(v))\leq F_n= \pi(\psi(v^{(n+1)}))$, and the first part of theorem is proved.

As regards the second part, the result is trivial for $n\leq 1$. We shall suppose $n>1$ and
 prove that for $v\in{\cal A}^{n+1}$, $n\geq 1$, the maximal value of $\pi(\psi(v))$ is reached if and
only if $v$ is one of the following words $ v^{(n+1)}, E(v^{(n+1)})$, $ c(v^{(n+1)})$, and  $ E(c(v^{(n+1)}))$.

For what concerns the  `if part' of the statement we have proved above that $\pi(\psi(v^{(n+1)}))= \pi(E(\psi(v^{(n+1)})))=\pi(\psi(E(v^{(n+1)}))) = F_{n}$. Let us now prove that $$\pi(\psi(v^{(n+1)}))=\pi(\psi(c(v^{(n+1)}))). $$
Since $v^{(n+1)}= v^{(n)}z= v^{(n-1)}\bar z z$, one has $c(v^{(n+1)})= v^{(n-1)}z\bar z$. From Justin's
formula one derives:
$$\psi(c(v^{(n+1)}))=\psi( v^{(n-1)}z\bar z)= \mu_{v^{(n-1)}}(z\bar z)\psi(v^{(n-1)}z).$$
By Proposition \ref{prop:standStu},  $\mu_{v^{(n-1)}}(z\bar z)= \psi( v^{(n-1)})z\bar z$, so that
$$\psi(c(v^{(n+1)}))=\psi( v^{(n-1)})z\bar z \psi(v^{(n-1)}z).$$
From Proposition \ref{Prop:uno} and Lemma \ref{lem:fib000}, $\pi(\psi(c(v^{(n+1)})))= |\psi(v^{(n-1)})|+2= F_n=\pi(\psi(v^{(n+1)}))$.

Let us now prove the `only if part'. We suppose that $v\in {\cal A}^{n+1}$ is such that $\pi(\psi(v)) =
\pi(\psi(v^{(n+1)}))= F_n$. This implies by (\ref{eq:psiu}), $$|\psi(u')|= F_n-2 \ \mbox{ and} \ |u'|=n-1.$$
By Theorem \ref{thm:main00} this can occur if and only if  
$$ u' = v^{(n-1)}  \ \mbox{or} \ \  u' = E(v^{(n-1)}).$$ 
Let us recall that $v^{(n+1)}= v^{(n-1)}\bar z z$ and $v = u'\bar y y$. Suppose first $ u' = v^{(n-1)}$.
If $y=z$, we have $v=v^{(n-1)}\bar z z =  v^{(n+1)}$. If $y= \bar z$, then one has  $v=v^{(n-1)}z \bar z=
  c(v^{(n+1)})$. In the case $ u' = E(v^{(n-1)})$ one has $v= E(v^{(n-1)})\bar y y$. If $y= \bar z$, then $v = E(v^{(n-1)}\bar z z)= E(v^{(n+1)})$. If $y=z$, then $v = E(v^{(n-1)} z \bar z)=
E(c(v^{(n+1)}))$, which concludes the proof.
\end{proof}
\begin{example}{\em For $n=4$  the maximum value of the minimal period of central words of order 4
is $5= F_3$. It is reached with the directive words $abab$, $abba$, $baba$, and $baab$. The corresponding central words are respectively, $\psi(abab)=abaababaaba$, $\psi(abba)=ababaababa$,  $E(\psi(abab))$, and $E(\psi(abba))$.}
\end{example}
\begin{thm}\label{thm:pmain0} The minimal periods of the palindromic prefixes of any order of a characteristic Sturmian word $s$ are maximal
 if and only if $s=f$ or $s= E(f)$.
\end{thm}
\begin{proof} The proof follows the same lines of that of Theorem \ref{thm:main}. Let $s= \psi(y)$, with $y=y_1\cdots y_n\cdots$, $y_i\in {\cal A}$, $i\geq 1$, be any characteristic Sturmian word. By Theorem \ref{thm:pmain}, for any $n\geq 0$,
$$ \pi(\psi(y_1\cdots y_n)) \leq \pi(\psi(v^{(n)})) = \pi(\psi(E(v^{(n)}))), $$
where $v^{(n)}$  and $E(v^{(n)})$ are respectively the prefixes of $(ab)^{\omega}$ and of $(ba)^{\omega}$ of length $n$. Since  $\psi(v^{(n)})$ and $E(\psi(v^{(n)}))$ are respectively the palindromic prefixes of order $n$ of $f$ and $E(f)$,  the  `if part'  of theorem follows. 

Let now $s=\psi(y)$ be any characteristic Sturmian word such that for any $n$ and $v\in {\cal A}^n$,
 $\pi(\psi(y_1\cdots y_n)) \geq \pi(\psi(v))$. In particular, one has  $\pi(\psi(y_1\cdots y_n)) \geq \pi(\psi(v^{(n)}))$.
 By Theorem \ref{thm:pmain} it follows that for any $n\geq 0$
 \begin{equation}\label{eq:psivin}
 \pi(\psi(y_1\cdots y_n)) =  \pi(\psi(v^{(n)})),
 \end{equation}
 where the equality occurs if and only if for any $n$, $y_1\cdots y_n$ is one of the following
 words $ v^{(n)}, E(v^{(n)}), c(v^{(n)})$, and $E(c(v^{(n)}))$. We can suppose, without loss of generality, that $y_1=a$, i.e., $y\in a{\cal A}^{\omega}$. In this case equality (\ref{eq:psivin})
 implies that for each $n$
 $$ \mbox{either} \  \ y_1\cdots y_n = v^{(n)} \ \ \mbox{or} \ \  \ y_1\cdots y_n = c(v^{(n)}).$$
 Let us prove that the preceding equation implies that for all $n\geq 0$ one has $y_1\cdots y_n = v^{(n)}$. This is trivial for $n\leq 1$. For $n=2$, one has that $y_1y_2= v^{(2)}=ab$ or $y_1y_2=ba$.
 However, this second case cannot occur since $y_1=a$. Thus $y_1y_2= v^{(2)}$. Let us now prove
 by induction that if  $y_1\cdots y_n = v^{(n)}$ with $n\geq 2$, then  $y_1\cdots y_{n+1} = v^{(n+1)}$.
 Indeed, suppose by contradiction that  $y_1\cdots y_{n-1} y_ny_{n+1} = c( v^{(n+1)})$. This would imply
 $v^{(n+1)}= y_1\cdots y_{n-1}\bar y_n \bar y_{n+1}$, so that  $v^{(n)}= y_1\cdots y_{n-1}\bar y_n$
 which is absurd. Thus $y= (ab)^{\omega}$ and $s=f$. 
 
  If $y\in b{\cal A}^{\omega}$, one
 proves in a perfect similar way that $y =(ba)^{\omega}$, i.e., $s=E(f)$ and this concludes
 the proof.
\end{proof}

The following lemma relates the composition, i.e., the number of letters $a$
and $b$, of a proper Christoffel word $a\psi(v)b$ to the minimal period of $\psi(v^{\sim})$.

\begin{lemma}\label{lem:comp1} For any proper Christoffel word $w= a\psi(v)b$,
$$\pi(\psi(v^{\sim})) = \min\{|w|_a, |w|_b\}.$$ 
In particular, if $v\in a{\cal A}^*$, then
$$ \pi(\psi(v^{\sim})) = |\psi(v)|_b+1.$$
\end{lemma}
\begin{proof} In view of (\ref{eq:perpsi}) and statement 3.  of Theorem \ref {thm:Lynd}, one has
$$ \pi(\psi(v^{\sim})) = \min\{ p_a(v^{\sim}), p_b(v^{\sim})\}=  \min\{|w|_a, |w|_b\}.$$
 If $v\in a{\cal A}^*$, then $|w|_b <|w|_a$. Hence, in such a case
 $$ \pi(\psi(v^{\sim})) = |w|_b = |\psi(v)|_b+1. \qedhere$$
\end{proof}

 Let us denote by $d$ the operator in ${\cal A}^*$ defined as:
$d(\varepsilon)= \varepsilon$, $d(x)= x$ for any $x\in {\cal A}$, and for $v=xyu$ with $u\in{\cal A}^*$, $x,y\in {\cal A}$, $d(v)=d(xyu)= yxu$. Thus the operator $d$ acting on  words $v$ of length $\geq 2$ changes the prefix  $xy$ of $v$ of length 2 in $yx$. As it is readily verified the operator $d$
 is related to $c$ as follows: for any $v\in {\cal A}^*$, $d(v)= (c(v^{\sim}))^{\sim}$. Moreover,  $d$
 commute with $E$.
 
 \begin{thm}\label{thm:comp000} For any $n\geq 0$ and $v \in a{\cal A}^*$ of length $n$
 $$ |\psi(v)|_b\leq  |\psi(v^{(n)})|_b = F_{n-1}-1,$$
 where the equality holds if and only if $v= v^{(n)}$ or $v= E(d(v^{(n)}))$.
 \end{thm}
 \begin{proof} By Lemma \ref{lem:comp1} one has:
 $$ |\psi(v)|_b = \pi(\psi(v^{\sim}))-1 \ \mbox{and} \   |\psi(v^{(n)})|_b = \pi(\psi((v^{(n)})^{\sim}))-1.$$
By Theorem \ref{thm:pmain},
 $$\pi(\psi(v^{\sim})) \leq \pi(\psi(v^{(n)})) = F_{n-1}.$$
 Moreover, since  $(v^{(n)})^{\sim}$ is equal to $v^{(n)}$ if $n$ is odd and is equal to $E(v^{(n)})$
 if $n$ is even, by Theorem \ref{thm:pmain} one has
 $$\pi(\psi((v^{(n)})^{\sim})) = \pi(\psi(v^{(n)})).$$
 Hence,
  $$ |\psi(v)|_b = \pi(\psi(v^{\sim}))-1\leq F_{n-1}-1=|\psi(v^{(n)})|_b$$
  and the first part of the theorem is proved.
  
  Now $ |\psi(v)|_b =  |\psi(v^{(n)})|_b$ if and only if 
  $$ \pi(\psi(v^{\sim})) =  \pi(\psi(v^{(n)})).$$
  By Theorem \ref{thm:pmain} this occurs if and only if $v^{\sim}$ is one of the following words:
  $ v^{(n)}, E(v^{(n)}), c(v^{(n)})$, and  $E(c(v^{(n)}))$. We have to consider two cases:
  
  \vspace{2 mm}
  
  \noindent
  Case 1. $n$ is even. The word  $v^{(n)}$ terminates  with the  letter $b$, so that, as $v$ begins with the letter $a$, $v^{\sim}$ cannot be equal to $v^{(n)}$. Similarly, $v^{\sim}$ cannot be equal to
  $E(c(v^{(n)}))$. Indeed, $c(v^{(n)})$ terminates with the letter $a$ and  $E(c(v^{(n)}))$ with the letter $b$. This would imply that, $v$ will begin with the letter $b$ which is a contradiction.
  
  Now, as one easily verifies,  $v^{\sim} = E(v^{(n)})$ if and only if  $v= v^{(n)}$. Moreover,
   $v^{\sim} = c(v^{(n)})$ if and only if  $v= E(d(v^{(n)}))$.
  
    \vspace{2 mm}
  
  \noindent
  Case 2. $n$ is odd. The word  $v^{(n)}$ is a palindrome beginning and terminating with the letter $a$. Thus $E(v^{(n)})$ is also a palindrome terminating with the letter $b$. Thus, as $v$ begins with the letter $a$, $v^{\sim}$ cannot be equal to  $E(v^{(n)})$. Similarly, the word $c(v^{(n)})$ terminates with the letter $b$, so that  $v^{\sim}$ cannot be equal to $c(v^{(n)})$.
  
  Trivially, as $v^{(n)}$ is a palindrome, $v^{\sim}= v^{(n)}$ if and only if $v = v^{(n)}$. Finally, it is ready verified that $v^{\sim}= E(c(v^{(n)}))$ if and only if  $v= E(d(v^{(n)}))$.
  
  \vspace{2 mm}
  
  Hence, in conclusion the maximal value of $ |\psi(v)|_b $ is reached if and only if $v= v^{(n)}$ or $v= E(d(v^{(n)}))$.
   \end{proof}
   
   \begin{example} {\em For $n=5$ the central words of $a{\cal A}^*$ with a maximal number of $b$  have the directive words  $v^{(5)}=ababa$ and $E(d(v^{(5)}))= abbab$. One has $\psi(ababa)= abaababaabaababaaba$, $\psi(abbab)= ababaabababaababa$, and the number of $b$ is $7=F_5-1$}.
   \end{example}
   
   \begin{thm}\label{thm: Stuca} The only characteristic Sturmian words beginning with the letter $a$ whose palindromic prefixes of any order have the maximal number of occurrences of the letter $b$ are the Fibonacci word 
   $f= \psi((ab)^{\omega}$ and the word $g = \psi(ab^2(ab)^{\omega})$.
   \end{thm}
   
   \begin{proof} Let  $s=\psi(y)$ be any characteristic Sturmian word such that $y=y_1y_2\cdots y_n\cdots$, with $y_1=a$ and $y_i\in {\cal A}$ for $i>1$. Let us suppose that
    for any $n\geq 1$, 
   $$ |\psi(y_1y_2\cdots y_n)|_b =  |\psi(x_1x_2\cdots x_n)|_b$$
   where $x_1x_2\cdots x_n = v^{(n)}$ is the prefix of length $n$ of the word $(ab)^{\omega}$. 
  Setting $v= y_1y_2\cdots y_n $, from Theorem \ref{thm:comp000} the preceding equality can occur if and only if $v= v^{(n)}$ or $v= E(d(v^{(n)}))$, that is
  $$ v = x_1x_2x_3\cdots x_n  \ \mbox{or} \ \ v= x_1x_2x_2x_3\cdots x_{n-1}.$$
  For any $n>2$, if $y_1y_2\cdots y_n= v^{(n)}$, then $y_1y_2\cdots y_n y_{n+1}\neq E(d(v^{(n+1)}))$
  so that $y_1y_2\cdots y_n y_{n+1} = v^{(n+1)}$.
  Similarly, if $y_1y_2\cdots y_n= E(d(v^{(n)}))$, then $y_1y_2\cdots y_n y_{n+1}\neq v^{(n+1)}$
  so that $y_1y_2\cdots y_n y_{n+1} =E(d(v^{(n+1)}))$.
  
  Thus if $y_1y_2y_3 = v^{(3)}=aba$, then  $y= (ab)^{\omega}$ and $s= \psi((ab)^{\omega})= f$. If, on the contrary, $y_1y_2y_3 = E(d(v^{(3)}))= abb$, then $y= ab^2(ab)^{\omega}$ and 
  $s=\psi(ab^2(ab)^{\omega})$.
   \end{proof}
   
From the preceding theorem one derives the following extremal property of Fibonacci word.
   
   \begin{cor} \label{cor: Stuca} Fibonacci word is the unique characteristic Sturmian word $s$ whose directive word begins with $aba$, or equivalently $s$ begins with $abaa$,  such that its palindromic prefixes of any order have the maximal number of occurrences of the letter $b$.
   \end{cor}

\section{Arithmetization}\label{sec:five}

In this section we shall give an interpretation of the extremal properties satisfied by the palindromic prefixes of $f$
and $E(f)$ shown in the preceding section, in terms of continued fractions and more precisely of
continuants.

Any word $v\in {\cal A}^*$ can be uniquely represented as:
\[
	v=b^{\alpha_0}a^{\alpha_1}b^{\alpha_2}\cdots a^{\alpha_{m-1}}b^{\alpha_m},
\]
where $m$ is an even integer, $\alpha_i>0$, $i=1,\dots,m-1$, and $\alpha_0\geq 0$,
$\alpha_m\geq 0$.  We call the list $(\alpha_0, \alpha_1,\dots,\alpha_n)$, where $n=m$ if
$\alpha_m>0$ and $n=m-1$ otherwise, the
\emph{integral representation} of the word $v$.

  We can
identify the word $v$ with its integral representation and  write
$v\equiv(\alpha_0, \alpha_1,\dots,\alpha_n)$. One has:
$$ |v| = \sum_{i=0}^{n}|\alpha_i|.$$
For instance, the words $v_1= b^2aba^2$ and $v_2= a^3bab^2$ have the integral representations
$v_1\equiv (2,1,1,2)$ and $v_2\equiv (0,3,1,1,2)$.

If $v\in{\cal A}^{\omega}$ is the directive word of the characteristic word
$\psi(v)$, then  $v$ can be uniquely represented by 
$$v=b^{\alpha_0}a^{\alpha_1}b^{\alpha_2}\cdots ,$$
with $\alpha_0\geq 0$ and $\alpha_i>0$, $i>0$.
The infinite sequence $({\alpha_0},{\alpha_1}, {\alpha_2},\cdots, {\alpha_n}, \cdots)$ is called 
the {\em integral representation} of $v$. It has been proved in \cite{deluca} that if $\alpha_0=0$ then $({\alpha_1}, {\alpha_2},\cdots, {\alpha_n}, \cdots)$ coincides with the directive numerical sequence of the characteristic word $\psi(v)$. If $\alpha_0>0$, then the directive numerical sequence of $\psi(v)$
is $(0, \alpha_0, \alpha_1, \ldots, \alpha_n, \ldots)$.

The following important theorem holds (cf.\cite{BDL,BLRS}):
\begin{thm}\label{thm:cf} Let $w=aub$ be a proper Christoffel word with $u=\psi(v)$ and $(\alpha_0,\alpha_1,\ldots, \alpha_n)$, $n\geq 0$, be
the  integral  representation of  $v$. Then the slope $\eta(w)$ of $w$ is given by the continued fraction
$$[\alpha_0;\alpha_1, \ldots, \alpha_{n-1}, \alpha_n+1].$$
\end{thm}

We remark that in the case $n=0$ the preceding formula becomes $[\alpha_0+1]$, or, equivalently,
$[\alpha_0; 1]$.

\begin{example}{\em Let $v=a^2b^2a$. One has $w= a^3ba^2ba^3ba^2ba^2b$ and $\eta(w)=[0;2,2,2]$=$\frac{5}{12}$. If $v=ba^2b$, then $w= abababbababb$ and $\eta(w)=[1;2,2]$=$\frac{7}{5}$. If $v=b^3$, then $w=ab^4$ and $\eta(w)= \frac{4}{1}= [4]= [3;1]$.}
\end{example}

Let $[a_0;a_1, \ldots, a_n]$ be a continued fraction. As is well known (see, for instance, \cite{K}), for any $0\leq k\leq n$, the $k$-order convergent
$C_k= [a_0;a_1,\ldots,a_k]$ is given by the ratio $\frac{A_k}{B_k}$, where $(A_k)_{k\geq -1}$, $(B_k)_{k\geq -1}$ is a bisequence defined by 
$$ A_{-1}=1, \ A_0=a_0, \ B_{-1}=0, \ B_0=1$$
and
$$A_{k+1}= a_{k+1}A_k+ A_{k-1}, \ B_{k+1}= a_{k+1}B_k+B_{k-1},$$
for $0\leq k \leq n-1$. For any $k\geq 0$ the fraction $\frac{A_k}{B_k}$
is irreducible.

Let us now set for any $k\geq -1$,
$$ P_k = A_k+ B_k.$$
One has that  $P_{-1} =1$, $P_0= a_0+1$, and 
\begin{equation}\label{eq:cf01}
 P_{k+1}= a_{k+1}P_k+P_{k-1}, \ \ \mbox{for} \ k\geq 0.
 \end{equation}
The value of $P_n$ for $n\geq 0$ can be expressed in terms of {\em continuants} (cf.~\cite{GKP}) (called {\em cumulants} in \cite{EL}). Let  $a_0,a_1, \ldots, a_n, \ldots$ be any sequence of 
numbers. The $n$-th continuant  $K[a_0,\ldots, a_n]$ is defined recursively as: $K[ \ \ ]= 1$,
$K[a_0]= a_0$,
 and for $n\geq 1$,
\begin{equation}\label{eq:cf02}
 K[a_0,a_1,\ldots, a_n] = a_nK[a_0,a_1,\ldots, a_{n-1}]+ K[a_0,a_1,\ldots, a_{n-2}].
 \end{equation}
As it is ready verified for any $n\geq 0$, $K[a_0,a_1,\ldots, a_n]$ is a multivariate polynomial
in the variables $a_0, a_1, \ldots, a_n$ which is obtained by starting with the product $a_0a_1\cdots a_n$ and then striking out adjacent pairs  $a_ka_{k+1}$ in all possible ways. For instance, $K[a_0,a_1,a_2,a_3,a_4]= a_0a_1a_2a_3a_4+ a_2a_3a_4+a_0a_3a_4+a_0a_1a_4+a_0a_1a_2+a_0+a_2+a_4.$

We recall (cf.~\cite{GKP,EL})  that for every $n\geq 0$,
\begin{equation}\label{eq:reverse} 
K[a_0,\ldots, a_n]= K[a_n, \ldots, a_0],
\end{equation}
 i.e., a continuant does not change its value by reversing the order of its elements; moreover, one has
$K[1^n] = F_{n-1}$, where we have denoted by $1^n$ the sequence of length $n$, $(1,1, \ldots, 1)$. 
A further property that we shall use in the following,  is:
\begin{equation}\label{eq:piuno}
K[a_0,\ldots, a_n, 1] = K[a_0,\ldots, a_{n-1}, a_n+1].
\end{equation}
There exists a strong relation between continued fractions and continuants. More precisely the
following holds. Let $[a_0;a_1,\ldots, a_n]$ be any continued fraction. Then
\begin{equation}\label{eq:cfcont}
[a_0;a_1,\ldots, a_n]= \frac{K[a_0,a_1,\ldots, a_{n}]}{K[a_1,\ldots, a_{n}]}
\end{equation}
Indeed, as is ready verified, $K[a_0,a_1,\ldots, a_{n}]= A_n$ and $K[a_1,\ldots, a_{n}]=B_n$.

From (\ref{eq:cf01}) and (\ref{eq:cf02}), or using the preceding properties of continuants, one derives that if  $[a_0;a_1, \ldots, a_n]$ is a continued fraction,  then for any $n\geq 0$,
\begin{equation}\label{eq:cf3}
P_n = A_n+B_n=  K[a_0+1, a_1,\ldots, a_n].
\end{equation}
The following holds:

\begin{thm}\label{thm:cf000} Let $w=aub$ be a proper Christoffel word with $u=\psi(v)$ and $(\alpha_0,\alpha_1,\ldots, \alpha_n)$, $n\geq 0$, be
the  integral representation of  $v$. Then 
$$ |w|= K[\alpha_0+1,\alpha_1,\ldots, \alpha_{n-1}, \alpha_n+1].$$
\end{thm}
We remark that for $n=0$ the preceding formula becomes $K[\alpha_0+1,1]= K[\alpha_0+2]$.
\begin{proof} By Theorem \ref{thm:cf}, the slope $\frac{|w|_b}{|w|_a}$ of $w$ is given by the continued fraction 
$[\alpha_0;\alpha_1, \ldots, \alpha_{n-1}, \alpha_n+1].$ Since  the $n$-th order convergent $C'_n= \frac{A'_n}{B'_n}=\frac{|w|_b}{|w|_a}$ and $\gcd(|w|_a,|w|_b)=1$,
one has $P'_n= A'_n+B'_n = |w|_b+|w|_a= |w|$. Then the result follows from (\ref{eq:cf3}).
\end{proof}

Theorem \ref{thm:main00}   and Corollary \ref{cor:main}  can be restated equivalently in terms of continuants as follows:

\begin{thm}\label{thm:cf0101} Let $n\geq 0$ and $\alpha_0, \alpha_1, \ldots, \alpha_m$ be any sequence of integers such that 
$$\alpha_0\geq 0,  \ \alpha_i>0,  \ i=1,\ldots, m,\ \  \mbox{and} \  \  \sum_{i=0}^m \alpha_i= n.$$ Then
\begin{equation}\label{eq:cappa}
  K[\alpha_0+1,\alpha_1,\ldots, \alpha_{m-1}, \alpha_m+1] \leq K[1^n, 2] = K[2, 1^n] =F_{n+1},
\end{equation}
where the equality occurs if and only if $m=n$ and $\alpha_0=0, \alpha_1=\alpha_2= \cdots =\alpha_n=1$ or $m=n-1$ and  $\alpha_0= \alpha_1=\alpha_2= \cdots = \alpha_{n-1}=1$.
\end{thm}
\begin{proof} Let $v$ be any word of ${\cal A}^n$ having the integral representation
$v\equiv (\alpha_0, \alpha_1, \ldots, \alpha_m)$ such that
$n= \sum_{i=0}^m \alpha_i$. By  Theorem \ref{thm:main00}   and Corollary \ref{cor:main} one has
$$ |\psi(v)| \leq |\psi(v^{(n)})|= F_{n+1}-2,$$
so that $|a\psi(v)b| \leq F_{n+1}$. Thus, by Theorem \ref{thm:cf000} one derives
$$ |a\psi(v)b|= K[\alpha_0+1,\alpha_1,\ldots, \alpha_{m-1}, \alpha_m+1] \leq F_{n+1}= K[1^n, 2]=K[2, 1^n].$$
By Corollary \ref{cor:main} the equality occurs if and only if $v= v^{(n)}$ or  $v= E(v^{(n)})$. In the first case $m=n$ and $\alpha_0=0, \alpha_1=\alpha_2= \cdots =\alpha_n=1$. In the second case
$m=n-1$,  and $\alpha_0= \alpha_1=\alpha_2= \cdots =\alpha_{n-1}=1$. Hence, the theorem
is proved.
\end{proof}

Let us observe that the preceding theorem implies the validity of Theorem \ref{thm:main00}   and Corollary \ref{cor:main}. Indeed,
let $v$ be any word over ${\cal A}$ having the integral representation
$v\equiv (\alpha_0, \alpha_1, \ldots, \alpha_m)$ and length $n= \sum_{i=0}^m \alpha_i$. From
(\ref{eq:cappa}) and Theorem \ref{thm:cf000} one derives $|\psi(v)|\leq |\psi(v^{(n)})|= F_{n+1}-2,$
where the equality holds if and only if  $v= v^{(n)}$ or  $v= E(v^{(n)})$. 

We shall  give now  a direct proof of Theorem \ref{thm:cf0101} without using combinatorics on words.
We need the following lemma on Fibonacci numbers.

\begin{lemma}\label{lem:fib00} Let $n\geq 1$. For any  integer  $x$ such that  $0< x \leq n$, one has:
$$xF_{n-x}+ F_{n-x+1}\leq F_{n+1},$$
where the equality holds if and only if  $x=1$.
\end{lemma}
\begin{proof} The proof is by induction on the value of $x\leq n$.  For $x=1$ one has $F_{n-1}+F_n= F_{n+1}$. For $x=2\leq n$ one has
$2F_{n-2}+F_{n-1}= F_{n-2}+F_{n-2}+F_{n-1}= F_{n-2}+F_n < F_{n-1}+F_n= F_{n+1}$. Suppose the statement true up to $1<x-1<n$ and prove it for $x$. One has by using the inductive hypothesis,
$$xF_{n-x}+F_{n-x+1} = (x-1)F_{n-x}+ F_{n-x}+ F_{n-x+1}= (x-1)F_{n-x}+F_{n-x+2} $$
$$< (x-1)F_{n-x+1}+F_{n-x+2} <F_{n+1}. \qedhere$$  \end{proof}

\vspace{3mm}

\noindent
{\em (Second proof of Theorem \ref{thm:cf0101})}.
The proof is by induction on the integer $n$. The result is trivial if $n\leq 1$. Let us suppose the result true for all integers less than $n>1$ and prove it for $n$. Let $\alpha_0, \alpha_1, \ldots, \alpha_m$ be any sequence of integers such that 
$\alpha_0\geq 0,  \ \alpha_i>0,  \ i=1,\ldots, m, \ \  \mbox{and} \  \  \sum_{i=0}^m \alpha_i= n$. From the definition of continuant one has:
$$   K[\alpha_0+1,\alpha_1,\ldots, \alpha_{m-1}, \alpha_m+1]= (\alpha_m+1)  K[\alpha_0+1,\alpha_1,\ldots, \alpha_{m-1}] $$ $$ + K[\alpha_0+1,\alpha_1,\ldots, \alpha_{m-2}].$$
By induction one  derives:
\begin{equation}\label{eq:zeroemme1}
K[\alpha_0+1,\alpha_1,\ldots, \alpha_{m-1}]\leq F_{n-\alpha_m}. 
\end{equation}
Indeed, if $\alpha_{m-1}>1$ one has $$K[\alpha_0+1,\alpha_1,\ldots, \alpha_{m-1}] =K[\alpha_0+1,\alpha_1,\ldots, \alpha_{m-2}, (\alpha_{m-1}-1)+1].$$
Since, $\sum_{i=0}^{m-2} \alpha_i +(\alpha_{m-1}-1) = n-\alpha_m-1$, equation (\ref{eq:zeroemme1}) follows by induction. If $\alpha_{m-1}=1$, by (\ref{eq:piuno}), one has  
$$K[\alpha_0+1,\alpha_1,\ldots, \alpha_{m-2}, 1]= K[\alpha_0+1,\alpha_1,\ldots, \alpha_{m-2}+ 1].$$
Since $\sum_{i=0}^{m-2} \alpha_i = n-\alpha_m-1$, equation (\ref{eq:zeroemme1}) follows again by induction.
In a similar way one derives by induction
\begin{equation}\label{eq:zeroemme2}
K[\alpha_0+1,\alpha_1,\ldots, \alpha_{m-2}]\leq F_{n-\alpha_m -\alpha_{m-1}}. 
\end{equation}
Thus, since $\alpha_{m-1}\geq 1$, one has:
$$   K[\alpha_0+1,\alpha_1,\ldots, \alpha_{m-1}, \alpha_m+1]\leq (\alpha_m+1)F_{n-\alpha_m} +
 F_{n-\alpha_m -\alpha_{m-1}}$$
 $$= \alpha_m F_{n-\alpha_m}+F_{n-\alpha_m}+F_{n-\alpha_m -\alpha_{m-1}}\leq  \alpha_m F_{n-\alpha_m}+   F_{n-\alpha_m+1},$$
 where in the last inequality the equality sign occurs if and only if $\alpha_{m-1}=1$.
 By Lemma \ref{lem:fib00},  $\alpha_m F_{n-\alpha_m}+   F_{n-\alpha_m+1} \leq F_{n+1}$, where the equality holds if and only if $\alpha_m= 1$. Thus in any case
 $$  K[\alpha_0+1,\alpha_1,\ldots, \alpha_{m-1}, \alpha_m+1]\leq F_{n+1}. $$
The equality can occur in the preceding equation if and only if $\alpha_m= \alpha_{m-1}=1$ and, moreover, in view of (\ref{eq:zeroemme1}) and (\ref{eq:zeroemme2}),
 $$K[\alpha_0+1,\alpha_1,\ldots, \alpha_{m-2}+1] = F_{n-1} \ \mbox{and} \  K[\alpha_0+1,\alpha_1,\ldots, \alpha_{m-2}] = F_{n-2}.$$
 Since $\sum_{i=0}^{m-2}\alpha_i = n-2$, by induction the first of the two preceding equations is satisfied if and
 only if $\alpha_0=0$, $m=n$, and $\alpha_1= \cdots =\alpha_{n-2}=1$ or $\alpha_0=1$, $m=n-1$,
 and $\alpha_1= \cdots = \alpha_{n-3}= 1$. In the first case $\alpha_{m-1}=\alpha_{n-1}=\alpha_n= \alpha_m=1 $, and  in the second case, $\alpha_{m-1}=\alpha_{n-2}=\alpha_m= \alpha_{n-1}=1$.
 Since for the previous values of $\alpha$'s the second equation is certainly satisfied, the result follows. $\Box$
 
 \begin{prop}\label{prop:ultima0} Let $v\in {\cal A}^*$ be a word having the integral representation $v=(\alpha_0,\alpha_1,\ldots, \alpha_n)$. Then
 $$ \pi(\psi(v)) = K[\alpha_0+1, \alpha_1, \ldots, \alpha_{n-1}].$$
 \end{prop}
 \begin{proof} It has been proved in \cite{deluca} (see also \cite{BDL,CdL}) that if $v$ has the
 integral representation $v=(\alpha_0,\alpha_1,\ldots, \alpha_n)$, then 
  $$ [0; \alpha_n, \alpha_{n-1}, \ldots, \alpha_1, \alpha_0+1] = \frac{\pi(\psi(v))}{q},$$
  where $\pi(\psi(v))$ is the minimal period of $\psi(v)$ and $q$ is the period of $\psi(v)$ such that $\gcd (q, \pi(\psi(v)))=1$ and $|\psi(v)|= \pi(\psi(v))+q -2$. By (\ref{eq:cfcont}) one has:
   $$ [0; \alpha_n, \alpha_{n-1}, \ldots, \alpha_1, \alpha_0+1] =  \frac{K[0,\alpha_n,\ldots, \alpha_1, \alpha_0+1]}{K[\alpha_n,\ldots, \alpha_1, \alpha_{0}+1]}.$$
   Since the preceding fraction is irreducible, by (\ref{eq:reverse}) and (\ref{eq:cf02}) one derives:
    $$ \pi(\psi(v)) = K[0, \alpha_n, \ldots, \alpha_1, \alpha_0+1] =K[\alpha_0+1, \alpha_1, \ldots, \alpha_{n-1}, \alpha_n, 0]$$
    $$ = K[\alpha_0+1, \alpha_1, \ldots, \alpha_{n-1}],$$
    which concludes the proof.
 \end{proof}
 
 By the preceding proposition and the extremal property of continuants expressed by Theorem \ref{thm:cf0101},  we can give a different proof of Theorem \ref{thm:pmain}. Indeed, the following
 proposition holds:
 
 \begin{prop}\label{prop:ultima1} Let $n\geq 0$ and $\alpha_0, \alpha_1, \ldots, \alpha_m$ be any sequence of integers such that 
$$\alpha_0\geq 0,  \ \alpha_i>0,  \ i=1,\ldots, m,\ \  \mbox{and} \  \  \sum_{i=0}^m \alpha_i= n.$$ 
One has that
\begin{equation}\label{eq:ultimo}
  K[\alpha_0+1,\alpha_1,\ldots, \alpha_{m-1}] \leq F_{n-1}.
  \end{equation}
  The equality is reached if and only if one of the following conditions is satisfied:
 
 \vspace{2 mm}
 
 \noindent
 1) $\alpha_0=0$, $m=n$, and $\alpha_1=\alpha_2= \cdots =\alpha_{n-1}=\alpha_n=1$, 
 
 \noindent
 2) $\alpha_0=0$, $m=n-1$, and $\alpha_i=1$ for $ 1\leq i \leq n-3$, $\alpha_{n-2}=2$, $\alpha_{n-1}=1$,
 
 \noindent
 3) $\alpha_0=1$, $m=n-1$, and $\alpha_1=\alpha_2= \cdots =\alpha_{n-1}=1$,
 
 \noindent
 4) $\alpha_0=1$, $m=n-2$, and $\alpha_i=1$ for $ 1\leq i \leq n-4$, $\alpha_{n-3}=2$,  $\alpha_{n-2}=1$.
 \end{prop}
 \begin{proof} We have to consider two cases. If  $\alpha_{m-1}=1$, 
 since  $$K[\alpha_0+1,\alpha_1,\ldots, \alpha_{m-2}, 1]= K[\alpha_0+1,\alpha_1,\ldots, \alpha_{m-2} +1],$$ one has
 \begin{equation} \label{eq:prima00}
 K[\alpha_0+1,\alpha_1,\ldots, \alpha_{m-2} +1] \leq F_{n-\alpha_m} \leq F_{n-1}.
 \end{equation}
 Indeed,  since  $\sum_{i=0}^{m-2} \alpha_i= n-1-\alpha_m$, the preceding formula follows from
 Theorem \ref{thm:cf0101}. 
 If $\alpha_{m-1}>1$, one derives:
 \begin{equation}\label{eq:seconda00}
 K[\alpha_0+1,\alpha_1,\ldots, \alpha_{m-2}, (\alpha_{m-1}-1)+1] \leq F_{n-\alpha_m} \leq F_{n-1}.
 \end{equation}
 Indeed, since $\sum_{i=0}^{m-2} \alpha_i+ \alpha_{m-1}-1=n-1-\alpha_m$, the previous inequality follows again from  Theorem \ref{thm:cf0101}. Thus in any case (\ref{eq:ultimo}) is satisfied.
 
 The maximal value of  $K[\alpha_0+1,\alpha_1,\ldots, \alpha_{m-1}]$ is then $F_{n-1}$.
 It is reached if and only if  one of the conditions  $1), 2), 3),$ and $4)$ is satisfied. The sufficiency of
 the preceding conditions is readily verified. Let us prove the necessity.
 
 Indeed, necessarily $\alpha_m=1$. Moreover, by Theorem \ref{thm:cf0101}, if $K[\alpha_0+1,\alpha_1,\ldots, \alpha_{m-1}]=F_{n-1}$,
  then   $\alpha_0=0$ or $\alpha_0=1$. We  consider  only  the case $\alpha_0=0$; the case
  $\alpha_0=1$ is similarly dealt with. 
  
  If  $\alpha_{m-1}=1$, in view of (\ref{eq:prima00}), one derives by Theorem \ref{thm:cf0101}, that $m=n$ and  $\alpha_1=\alpha_2= \cdots =\alpha_{n-2}=1$. Since $\alpha_{m-1}=\alpha_{n-1}=\alpha_m=\alpha_n=1$,   condition 1)  is satisfied.
  
   If  $\alpha_{m-1}>1$, in view of (\ref{eq:seconda00}), one derives by Theorem \ref{thm:cf0101}, that
   $m=n-1$ and $\alpha_1=\alpha_2= \cdots =\alpha_{m-2}= \alpha_{m-1}-1=1$. Thus $\alpha_{m-1}=\alpha_{n-2}=2$.  Hence, since $\alpha_m = \alpha_{n-1}=1$,  condition 2)  is satisfied.
 \end{proof}
By Propositions \ref{prop:ultima0} and \ref{prop:ultima1}, one easily derives  Theorem \ref{thm:pmain} of the previous section.

\end{document}